\documentclass[%
 reprint,
 amsmath,amssymb,
 aps,
floatfix,
pre
]{revtex4-2}

\usepackage{graphicx}% Include figure files
\usepackage{dcolumn}% Align table columns on decimal point
\usepackage{bm}% bold math
\usepackage{amsfonts}
\usepackage{amsmath}
\usepackage{amsthm}
\newtheorem{theorem}{Theorem}

\usepackage{changes}

\begin{document}

\preprint{APS/123-QED}

\title{Minimal-Dissipation Learning for\\Energy-Based Models}% Force line breaks with \\

\author{Jeff~Hnybida}
 \email{Corresponding author: jeffrey.hnybida@irreversible.tech}
\author{Simon~Verret}%
\affiliation{%
 Irr\'eversible Inc., Sherbrooke, QC, Canada
}%

\date{\today}% It is always \today, today,

\begin{abstract}
      We show that the bias of the approximate maximum-likelihood estimation (MLE) objective of a persistent chain energy-based model (EBM) is precisely equal to the thermodynamic excess work of an overdamped Langevin dynamical system. We then answer the question of whether such a model can be trained with minimal excess work, that is, energy dissipation, in a finite amount of time. We find that a Gaussian energy function with constant variance can be trained with minimal excess work by controlling only the learning rate.   This proves that it is possible to train a persistent chain EBM in a finite amount of time with minimal dissipation and also provides a lower bound on the energy required for the computation. We refer to such a learning process that minimizes the excess work as \emph{minimal-dissipation learning}. We then provide a generalization of the optimal learning rate schedule to general potentials and find that it induces a natural gradient flow on the MLE objective, a well-known second-order optimization method.
\end{abstract}

\maketitle

%\tableofcontents

\section{Introduction}

Successful techniques in machine learning have often borrowed ideas from statistical physics and thermodynamics. Early examples include the Hopfield network~\cite{hopfield_Neural_1982} and the Boltzmann \mbox{machine~\cite{ackley_Learning_1985, hinton_Training_2002, salakhutdinov_Restricted_2007},} both of which are examples of a more general class of models called energy-based models (EBM)~\cite{lecun2006tutorial}. A more recent example of machine learning emulating statistical physics is denoising diffusion probabilistic models (DDPM)~\cite{sohl-dickstein_Deep_2015, ho_Denoising_2020, song_ScoreBased_2020}, which were inspired by the Jarzynski equality~\cite{jarzynski_Equilibrium_1997, jarzynski_Nonequilibrium_1997, neal_Annealed_2001}, and became the dominant approach in the generative modelling of images. Aside from providing mere inspiration, these analogies with statistical physics invite the calculation of thermodynamic quantities such as entropy production associated with these machine learning algorithms, which in turn bounds the energy required for computation, as in the case of Landauer's principle for digital computation~\cite{landauer1961irreversibility}.

To reach the energy-efficiency bounds for computation set by thermodynamics, it seems necessary to abandon the digital von Neumann architecture and to consider unconventional computers. There have been recent proposals~\cite{coles_Thermodynamic_2023, scellier_Energybased_2023} claiming that analog computers could run EBMs and DDPMs orders of magnitude more efficiently than traditional digital hardware. True analog random processes (i.e., those that are not merely simulated) would execute the  prohibitively computationally expensive digital Monte Carlo algorithms. The use of intrinsic random processes to simulate Langevin Monte Carlo algorithms is referred to as ``thermodynamic computing''~\cite{conte_Thermodynamic_2019, melanson_Thermodynamic_2025} and raises the question of what theoretical minimum energy is required to run these machine learning algorithms.

In this paper, we derive lower bounds on the energy required to train a particular class of generative models: persistent chain EBMs~\cite{du2019implicit, nijkamp2020anatomy}. These models served as a stepping stone towards the development of DDPMs. They use convolutional neural networks~\cite{lecun1989backpropagation} as their energy function and Langevin Monte Carlo techniques to estimate the training objective. The term ``persistent chain'' refers to the use of a replay buffer in which Monte Carlo samples are stored and updated throughout a training process. This important difference compared to other EBM training techniques~\cite{ackley_Learning_1985, hinton_Training_2002, salakhutdinov_Restricted_2007, nijkamp_Learning_2019, grathwohl2019your, agoritsas_Explaining_2023} makes persistent chain EBMs the only \emph{convergent} EBMs~\cite{nijkamp2020anatomy} in the sense that the samples generated by the EBM are distributed according to the Gibbs distribution defined by the energy function.

The first contribution of our work is in connecting the training process of persistent chain EBMs~\cite{nijkamp2020anatomy} to stochastic thermodynamics~\cite{sekimoto1997complementarity, sekimoto_Stochastic_2010, seifert2012stochastic, peliti_Stochastic_2021}, using an experimentally realizable harmonic trap~\cite{schmiedl2007optimal} as a guiding example. Several valuable insights in understanding machine learning algorithms~\cite{agoritsas_Explaining_2023} or developing new ones~\cite{carbone2024efficient, klinger_Universal_2024, premkumar_Neural_2024} were recently derived by drawing a connection with stochastic thermodynamics. In our work on persistent chain EBMs, we assume that both the learning updates and the Monte Carlo updates in the EBM training process are continuous in time and that sample averages are exact. Our goal is to obtain insight regarding the general learning dynamics of EBMs. We use the harmonic trap as a working example for every general result we present, and we extend various exact results from stochastic thermodynamics~\cite{schmiedl2007optimal, sivak2012thermodynamic} to the case of EBMs.

Another contribution we make is in identifying the bias in the maximum-likelihood estimation (MLE) training objective with the thermodynamic excess work. This bias is a result of running the Monte Carlo algorithm to approximate the MLE objective.  This association of the bias with the excess work is valid for persistent chain EBMs with arbitrary energy functions. Minimizing the bias thus corresponds to minimizing the excess work, the latter of which is also a sought after goal in stochastic thermodynamics~\cite{aurell2011optimal}.

In~\cite{salazar2017nonequilibrium} the authors give a stochastic thermodynamic interpretation of the unsupervised learning of Restricted Boltzmann Machines (RBM).  An RBM can be viewed as an EBM on a discrete space using contrastive divergence (CD) loss, which is a different approximation of the MLE loss than the one we consider.  Using this thermodynamic interpretation they show that the CD loss can be expressed in terms of thermodynamic quantities such as entropy and heat and they show that in the limit of many Gibbs sampling steps, the CD loss (which becomes the MLE loss in the limit) equals the excess work.  Since they assume that the initial and final states are in equilibrium, the excess work they consider is equal to the total entropy production.  This connection between the MLE loss and entropy production follows from the generalized De Brujin identity~\cite{wibisono2017information}.  In contrast, we identify the excess work with the bias of the approximate MLE loss, not with the loss itself, and the excess work is not assumed to be equal to the total entropy production.

Our third contribution is in developing a general framework for minimizing excess work through learning rate scheduling.  While optimal protocols (i.e., time parametrizations of model parameters which minimize excess work) for the harmonic trap are well-known~\cite{schmiedl2007optimal, sivak2012thermodynamic}, we derive learning rate schedules for the MLE gradient flow of EBM training for which parameters follow these optimal protocols. We refer to these learning rate schedules as \emph{minimal-dissipation learning}. Following Refs.~\cite{sivak2012thermodynamic, amari2016information}, we obtain a generalization of the optimal learning schedules found for the harmonic trap. We find that these optimal learning schedules induce a natural gradient flow on the MLE objective, which is a well-known second-order optimization method~\cite{martens2020new}.  This provides a new connection between stochastic thermodynamics and information geometry~\cite{amari2016information}. 

It should be noted that the terms ``excess work'' and ``efficiency'' as used in the context of this paper correspond to lower bounds for computational efficiency and reduced power consumption for ideal computers such as a hypothetical version of the analog hardware mentioned above. These theoretical lower bounds for energy consumption are analogous to the Landauer bound for digital computers~\cite{landauer1961irreversibility}, and more specifically the bounds for finite-time protocols~\cite{aurell2012refined}. Hence, the optimal protocols presented here are not expected to be optimal for a digital computer simulating a discretized Langevin system, as in that case additional overhead must be taken into account, but also complete information about the thermodynamic system is known by the simulation, which could be used to reduce the excess work potentially to zero~\cite{parrondo2015thermodynamics}. Nevertheless, our results are also relevant with respect to digital computers, since the learning trajectories obtained follow a natural gradient flow in the parameter space, which are known to be asymptotically Fisher efficient in the case of maximum-likelihood tasks~\cite{amari1998natural}, and could serve as a guide in the development of more-efficient EBM training techniques.

The paper is organized as follows. In Section~\ref{sec_ebm}, we review the theory of persistent chain EBMs and show how the harmonic trap from stochastic thermodynamics can be interpreted as one. In Section~\ref{sec_dissipation}, we show that the bias introduced by the approximate MLE objective used to train persistent chain EBMs can be identified with the excess work of the EBM training process. We then derive learning rate schedules that minimize excess work, in the case where the parameters are assumed to be continuous at both the beginning and the end of the training process, and also when the parameters are allowed to have discontinuities at these endpoints. In Section~\ref{sec_general}, we provide a generalization of the optimal learning rate schedules found for the harmonic trap, and prove that it induces a natural gradient flow.
%------------------------------------------------

\section{Energy-Based Models}
\label{sec_ebm}

An energy-based model (EBM) is defined by a real-valued function $E(\boldsymbol{x}, \boldsymbol{\theta})$, referred to as the \textit{energy} of a probabilistic state space. In this paper, we use boldface symbols to represent vectors (and matrices), where \mbox{$\boldsymbol{x} = (x_1, x_2, ..., x_d)^T$} is an element of a given sample space, and $\boldsymbol{\theta} = (\theta_1, \theta_2, ..., \theta_n)^T$ is an element of a set of real-valued model parameters.

The energy $E(\boldsymbol{x}, \boldsymbol{\theta})$ represents the probabilistic likelihood of $\boldsymbol{x}$ given $\boldsymbol{\theta}$, where a lower energy value corresponds to a higher likelihood. This defines a parametric probability distribution known as the Boltzmann distribution
\begin{equation}
\label{eqn_boltzmann}
  p_\text{m}(\boldsymbol{x}, \boldsymbol{\theta}) \equiv \frac{e^{-\beta E(\boldsymbol{x}, \,\boldsymbol{\theta})}}{Z(\boldsymbol{\theta})},
\end{equation}
where the subscript ``m'' stands for ``model'', $\beta$ is the inverse temperature, and
\begin{equation}
\label{eqn_partition_function}
    Z(\boldsymbol{\theta}) = \int e^{-\beta E(\boldsymbol{x}, \boldsymbol{\theta})} \, d\boldsymbol{x}
\end{equation} 
is a generally intractable normalization factor known as the \textit{partition function}. For notational simplicity, the dependence of all functions on $\beta$ is left implicit here and throughout the manuscript.

In principle, an EBM can be used to model any arbitrary distribution by applying maximum-likelihood estimation (MLE). This can be done via a gradient flow \cite{villani2009optimal} given by
\begin{equation}
\label{eqn_MLE_descent}
  \frac{d\boldsymbol{\theta}}{dt} = - \eta \nabla_{\boldsymbol{\theta}} \mathcal{L}_\text{MLE}(\boldsymbol{\theta}),
\end{equation}
given the negative log-likelihood $\mathcal{L}_\text{MLE}$ of a set of parameters $\boldsymbol{\theta}$, some learning rate $\eta$, and $\nabla_{\boldsymbol{\theta}} \equiv (\partial/\partial_{\theta_1}, ..., \partial/\partial_{\theta_n})$.

Now, suppose we want to model a distribution of data $p_\text{d}$, from which we assume we are able to generate independent and identically distributed (IID) samples. Then, the negative log-likelihood of the EBM's parameters $\boldsymbol{\theta}$ can simply be expressed as
\begin{equation}
\label{eqn_MLE_log_likelihood}
   \mathcal{L}_\text{MLE}(\boldsymbol{\theta}) \equiv \big< -\beta^{-1}\log p_\text{m} \big>_{p_\text{d}},
\end{equation}
where $\big< f \big>_{p} \equiv \int p(x) f(x) dx$ denotes the expectation of a function $f$ with respect to a probability distribution $p$, where $dx$ is the Lebesgue measure on the corresponding sample space. Unfortunately, Eq.~(\ref{eqn_MLE_log_likelihood}) is intractable due to the presence of the partition function $Z(\boldsymbol{\theta})$ in $\log p_\text{m}$. Nevertheless, to apply MLE, we require only the gradient
\begin{equation}
\label{eqn_MLE_gradient}
  \nabla_{\boldsymbol{\theta}} \mathcal{L}_\text{MLE}(\boldsymbol{\theta}) = \big< \nabla_{\boldsymbol{\theta}} E \big>_{p_\text{d}} - \big< \nabla_{\boldsymbol{\theta}} E \big>_{p_\text{m}},
\end{equation}
which can be estimated using the sample average approximation, assuming we are able to generate IID samples from both $p_\text{d}$ and $p_\text{m}$.

\subsection{Markov Chain Monte Carlo Methods}

Sampling from $p_\text{m}$ can usually only be approximated using Markov chain Monte Carlo (MCMC) methods.  For these algorithms, a set of samples are drawn from a prior distribution and evolved according to a Markov chain update rule. This results in a time-dependent \textit{sample distribution} $p_\text{s}$ which converges to $p_\text{m}$ as the number of iterations of the Markov chains approaches infinity. For any amount of finite time, however, the MCMC sample distribution does not match the model distribution:
\begin{equation}
p_\text{s} \neq p_\text{m}.
\end{equation}

The most common MCMC method used for EBMs is the Langevin Monte Carlo method, which is a discretization of Langevin dynamics given by
\begin{equation}
\label{eqn_langevin}
   \frac{d\boldsymbol{x}}{dt} = -\mu \nabla_{\boldsymbol{x}} E(\boldsymbol{x}, \boldsymbol{\theta}) + \boldsymbol{\xi},
\end{equation}
where $\mu$ represents the \textit{mobility} of the Langevin process and $\boldsymbol{\xi}$ is Brownian white noise with a mean of zero and a variance of $2 \mu \beta^{-1}$.

Samples from $p_\text{m}$ are thus generated by initializing samples from a prior distribution and evolving them according to Eq.~(\ref{eqn_langevin}) for an infinite amount of time.

Since, in practice, samples can only be generated by evolving a sample distribution for a finite amount of time, the MLE gradient used to train an EBM is instead approximated by
\begin{equation}
\label{eqn_approx_gradient}
  \nabla_{\boldsymbol{\theta}} \mathcal{L}_\text{approx}(\boldsymbol{\theta}) \equiv  \big< \nabla_{\boldsymbol{\theta}} E \big>_{p_\text{d}} - \big< \nabla_{\boldsymbol{\theta}} E \big>_{p_\text{s}},
\end{equation}
where the expectation over $p_\text{m}$ in Eq.~(\ref{eqn_MLE_gradient}) is replaced by an expectation over $p_\text{s}$. In doing so, the approximate training method is defined as
\begin{equation}
\label{eqn_approx_descent}
  \frac{d\boldsymbol{\theta}}{dt} = - \eta \nabla_{\boldsymbol{\theta}} \mathcal{L}_\text{approx}(\boldsymbol{\theta}).
\end{equation}
This method is practical for some real-world applications, but is extremely computationally expensive due to the long evolution times (corresponding to long MCMC chains) required to approximate Eq.~(\ref{eqn_langevin}) effectively.

One method of reducing the computation required for sampling from an EBM, which has been proven successful, is to run very short MCMC chains~\cite{hinton_Training_2002, salakhutdinov_Restricted_2007}, and periodically re-initialize the samples from the prior distribution or the distribution of the dataset~\cite{du2019implicit, grathwohl2019your, nijkamp_Learning_2019}. However, it has been shown~\cite{nijkamp2020anatomy} that samples obtained from these models do not converge to the data distribution in the limit of a large number of MCMC iterations and hence are not true EBMs of the data. Indeed, the model $p_\text{m}$, which is the equilibrium distribution of the Langevin dynamics, does not match the data distribution, that is, $p_\text{m}\neq p_\text{d}$.

For this reason, the focus of this paper is on persistent chain EBMs~\cite{nijkamp2020anatomy}, for which the sample distribution is initialized only once at the start of training, stored in a replay buffer, and updated throughout training. More specifically, we study the dynamics defined by Eqs.~(\ref{eqn_langevin}) and (\ref{eqn_approx_descent}). In simple terms, the energy model regresses to the data due to the MLE update given by Eq.~(\ref{eqn_approx_descent}), while the MCMC samples relax to the model distribution due to Langevin dynamics given by Eq.~(\ref{eqn_langevin}). Note that the timescale associated with Eq.~(\ref{eqn_approx_descent}) is implicitly chosen to coincide with that of Eq.~(\ref{eqn_langevin}), however there is a time rescaling invariance.  If one regards the timescale set by $\mu$ to be fundamental, for example as determined by a physical system, then $\eta$ can be regarded as a scale factor in which only the ratio $\eta/\mu$ is physically meaningful.

We study the continuous-time dynamics of EBMs, in which the model and the samples are updated continually, in contrast to discrete-time dynamics, in which the model and samples are updated one after another. Nevertheless, the training process ends when $p_\text{m}=p_\text{d}$, however this is a simplifying assumption and not necessarily true for EBM training in general. Note that this does not strictly require that $p_\text{s}=p_\text{m}=p_\text{d}$, as we demonstrate below in the case of minimal-dissipation learning.

\subsection{The Harmonic Trap as an Energy-Based Model}
\label{sec_harmonic_trap_ebm}

Here, we show how a one-dimensional, overdamped Langevin system, coupled to a time-dependent \textit{harmonic trap} can be viewed as a persistent chain EBM. The harmonic trap is an often-used, experimentally realizable example from stochastic thermodynamics~\cite{schmiedl2007optimal} which we will revisit several times in our analysis. The energy function of a harmonic trap is quadratic, which we express as
\begin{equation}
\label{eqn_harmonic_trap_energy}
  E(x, \, \theta) = \frac{1}{2}(x - \theta)^2,
\end{equation}
so $p_\text{m}$ is a normal distribution with a mean of $\theta$ and a constant variance of $\beta^{-1}$.

Suppose the data distribution $p_\text{d}$ has a centre of mass at $\theta^\ast = \langle x \rangle_{p_\text{d}}$. For simplicity, we can assume that
\begin{equation}
  p_\text{d}(x) \equiv p_\text{m}(x, \theta^\ast).
\end{equation}
Defining $u \equiv \langle x \rangle_{p_\text{s}}$, that is, the centre of mass of $p_\text{s}$, we can use the Langevin equation~(\ref{eqn_langevin}) to derive the following equation of motion for $u$:
\begin{equation}
\label{eqn_u_eom}
   \frac{du}{dt} = -\mu(u - \theta).
\end{equation}
Furthermore, Eq.~(\ref{eqn_approx_descent}) provides the equation of motion for $\theta$
\begin{equation}
\label{eqn_harmonic_trap_MLE}
   \frac{d\theta}{dt} = -\eta(u - \theta^\ast),
\end{equation}
and Eqs.~(\ref{eqn_u_eom}) and~(\ref{eqn_harmonic_trap_MLE}) form a closed system that can be solved exactly.

Note that, in the above equations of motion, we have used $\langle \nabla_\theta E \rangle_{p_\text{d}} = \theta - \theta^\ast$, that is, this expectation value is exact. In practice, some noise will exist in the expectation value due to the sample sizes being finite; we thus ignore the noise in the interest of simplicity.

In this simple example, $\theta^\ast$ can be computed directly from $\langle x \rangle_{p_\text{d}}$, but the goal of considering this harmonic trap as an EBM is not to determine the value of $\theta$; rather, it is to gain insight regarding the general learning dynamics which would apply to more-complicated energy functions, for which performing this type of direct calculation of $\theta$ would be intractable.

The initial state of $p_\text{s}$ is chosen to be a normal distribution centred at some value $u_0$, and the model parameter is initialized to some value $\theta_0$.

Figure~\ref{fig_harmonic_trap_constant_learning_rate} shows a ridge plot generated based on the evolution of Eqs.~(\ref{eqn_langevin}) and~(\ref{eqn_harmonic_trap_MLE}) for specific parameter values (refer to the caption for the values). Notice that the paths of the centres of mass of both the sample distribution $p_\text{s}$ and the model distribution $p_\text{m}$ oscillate. Similar oscillations exist in persistent chain EBM training for images~\cite{nijkamp2020anatomy} and are believed to be beneficial for the training. We argue that oscillations make the learning process inefficient. In the section that follows, we quantify this inefficiency using methods from stochastic thermodynamics.
\begin{figure}[h]
  \includegraphics[width=0.48\textwidth]{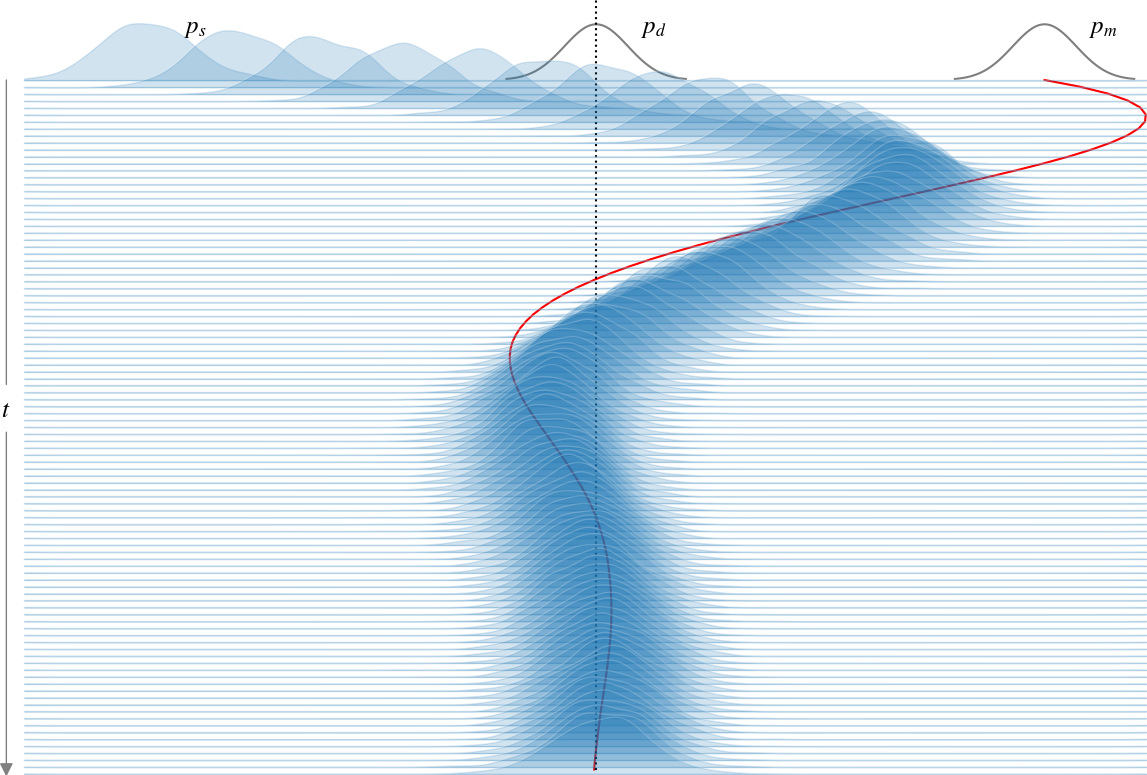}
\caption{The harmonic trap EBM is trained with a constant learning rate using the approximate MLE objective. The sample distribution, $p_\text{s}$, shown in blue using a ridge plot, evolves according to Langevin dynamics given by Eq.~(\ref{eqn_langevin}), whereas the model distribution $p_\text{m}$ having centre of mass $\theta$, shown by the red curve, is updated according to the MLE gradient flow described by Eq.~(\ref{eqn_harmonic_trap_MLE}) with a constant learning rate $\eta=10$. The data distribution $p_\text{d}$ is stationary, centred at $\theta^\ast$, as indicated by the vertical dotted line. The parameters used are $\theta_0=10$, $\theta^\ast=0$, $u_0=-10$, $\mu=10$, $\eta=10$, $\beta=1$, and $\tau=1$.}
\label{fig_harmonic_trap_constant_learning_rate}
\end{figure}
%

%------------------------------------------------

\section{Minimal-Dissipation Learning}
\label{sec_dissipation}

The approximation given by Eq.~(\ref{eqn_approx_gradient}) of the true MLE gradient in Eq.~(\ref{eqn_MLE_gradient}) is biased by the history of the evolution of the sample distribution. This bias can be mitigated by increasing the total time of the process, but doing so is undesirable in practice, for example, increasing the number of steps in a MCMC process necessitates additional computation to be performed.
  
Defining the bias as
\begin{equation}
\label{eqn_bias_MLE_approx}
  \mathcal{L}_\text{bias}(\boldsymbol{\theta}) \equiv \mathcal{L}_\text{MLE}(\boldsymbol{\theta}) - \mathcal{L}_\text{approx}(\boldsymbol{\theta}),
\end{equation}
and using Eqs.~(\ref{eqn_MLE_gradient}) and~(\ref{eqn_approx_gradient}), we find that
\begin{equation}
  \nabla_{\boldsymbol{\theta}} \mathcal{L}_\text{bias}(\boldsymbol{\theta}) =  \big< \nabla_{\boldsymbol{\theta}} E \big>_{p_\text{s}} - \big< \nabla_{\boldsymbol{\theta}} E \big>_{p_\text{m}}.
  \label{eqn_bias_def}
\end{equation}
The first term on the right-hand side can be identified with the average power output by a stochastic thermodynamic system which is defined by~\cite{sekimoto1998langevin}
\begin{equation}
  \frac{d \langle W \rangle_{p_\text{s}}}{dt} \equiv \big< \nabla_{\boldsymbol{\theta}} E \big>_{p_\text{s}} \cdot \frac{d\boldsymbol{\theta}}{dt},
  \label{eqn_work_def}
\end{equation}
since the persistent chain EBM is a Langevin dynamical system which satisfies energy conservation.  The second term on the right-hand side of Eq.~(\ref{eqn_bias_def}) corresponds to the gradient of the equilibrium free energy, which is defined by
\begin{equation}
\label{eqn_equilibrium_free_energy}
  F_\text{eq} \equiv -\beta^{-1} \log Z(\beta, \boldsymbol{\theta}).
\end{equation}
Substituting Eqs.~(\ref{eqn_work_def}) and (\ref{eqn_equilibrium_free_energy}) into Eq.~(\ref{eqn_bias_def}), we obtain
\begin{align}
   \frac{d\mathcal{L}_\text{bias}}{dt}
   &= \left(\big< \nabla_{\boldsymbol{\theta}} E \big>_{p_\text{s}} - \big< \nabla_{\boldsymbol{\theta}} E \big>_{p_\text{m}} \right) \cdot \frac{d{\boldsymbol{\theta}}}{dt} \label{eqn_bias_grad_E}, \\
   &= \frac{d \langle W \rangle_{p_\text{s}}}{dt} -  \frac{dF_\text{eq}}{dt}, \\
   &= \frac{d}{dt} \langle W_\text{ex} \rangle_{p_\text{s}}, \label{eqn_L_bias}
\end{align}
where $W_\text{ex}$ represents the \textit{excess work}
\begin{equation}
   \langle W_\text{ex} \rangle_{p_\text{s}} \equiv  \langle W \rangle_{p_\text{s}} - \Delta F_\text{eq}. 
\end{equation}
Identifying the bias of the approximate MLE objective to the thermodynamic quantity of excess work (up to a time independent constant) in this way is the first contribution of our research.

The change in equilibrium free energy $\Delta F_\text{eq}$ corresponds to the minimal amount of work required to transform a thermodynamic system between specified initial and final equilibrium states. The excess work, being the work in excess of this minimal amount, is therefore a measure of the total energy dissipation of the process.

It is evident from Eq.~(\ref{eqn_bias_grad_E}) that, if $p_\text{s} \approx p_\text{m}$, then the excess work and hence the bias can be minimized. This is what is known as a \textit{quasi-static process}, which can be achieved by initializing the system in equilibrium and updating the parameters $\theta(t)$ sufficiently slowly to ensure that $p_\text{s}$ evolves faster than $p_\text{m}$, thereby ensuring the system always remains close to equilibrium. 

Figure~\ref{fig_harmonic_trap_quasi_static} provides a visualization of the training process of the harmonic trap EBM in the quasi-static limit, that is, with a constant learning rate and $\eta \ll \mu$.  This is, in fact, the trajectory for the exact learning objective $\mathcal{L}_\text{MLE}$.

However, a quasi-static process can typically only produce zero excess work in the limit of infinite time, that is, if it is a reversible process. To minimize the excess work in a finite amount of time, we require a time-dependent protocol for $p_\text{m}$. In our analysis, we take the learning rate $\eta$ in Eq.~(\ref{eqn_MLE_descent}) to be the only controllable parameter so that the MLE objective is not altered.
\begin{figure}
  \includegraphics[width=0.48\textwidth]{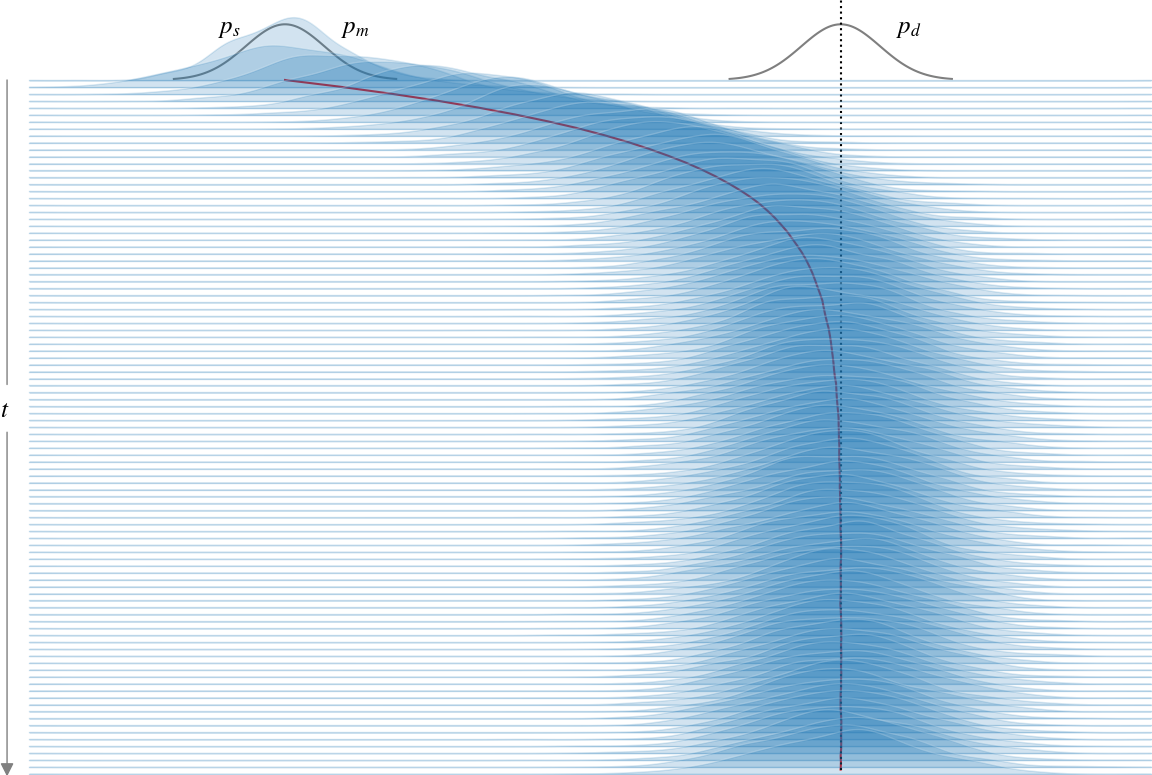}
\caption{The harmonic trap EBM is trained using the approximate MLE objective with a quasi-static protocol, i.e., with a constant learning rate $\eta \ll \mu$ so that $p_\text{s} \approx p_\text{m}$ at all times. The parameters chosen are $\theta_0=-10$, $\theta^\ast=0$, $u_0=-10$, $\mu=100$, $\eta=10$, $\beta=1$, and $\tau=1$.}
\label{fig_harmonic_trap_quasi_static}
\end{figure}
\subsection{Minimizing Dissipation for the Harmonic Trap}

The excess work of a one-dimensional harmonic trap, having either a fixed mean or fixed variance, was first derived in Ref.~\cite{schmiedl2007optimal}. In Section~\ref{sec_harmonic_trap_ebm} we recast this system as an EBM, raising the question of whether it is possible to learn the unknown parameter $\theta^\ast$ while controlling only the learning rate.

One important difference between our EBM and the harmonic trap in Ref.~\cite{schmiedl2007optimal} is that the sample distribution employed therein is always initialized in equilibrium, that is, $u_0 = \theta_0$, whereas we do not restrict the value of $u_0$ because MCMC samples of EBM training are generally not initialized in equilibrium. 

The excess work corresponding to the energy function~(\ref{eqn_harmonic_trap_energy}) has the form 
\begin{align}
  \langle W_\text{ex} \rangle_{p_\text{s}} 
  &= \int_{0}^{\tau} \frac{d\langle W_\text{ex} \rangle_{p_\text{s}}}{dt} dt, \\
  &= \int_{0}^{\tau}  \left(\big< \nabla_\theta E \big>_{p_\text{s}} - \big< \nabla_\theta E \big>_{p_\text{m}} \right) \dot{\theta} dt, \\
  &= \int_{0}^{\tau} (\theta - u) \dot{\theta} dt \label{eqn_excess_work_theta_minus_u},\\
  &= \int_{0}^{\tau} \frac{\dot{u}}{\mu} \left(\frac{\ddot{u}}{\mu} + \dot{u} \right) dt,
  \label{eqn_excess_work_all_u} \\
  &= \frac{1}{2}\left(\frac{\dot{u}}{\mu}\right)^2\bigg|_{0}^{\tau} +  \frac{1}{\mu} \int_{0}^{\tau} \dot{u}^2 dt, \label{eqn_excess_work_integral}
\end{align}
where Eq.~(\ref{eqn_u_eom}) is used to derive Eq.~(\ref{eqn_excess_work_all_u}) from Eq.~(\ref{eqn_excess_work_theta_minus_u}).

\subsubsection{Continuous Protocols}

To minimize excess work, we consider the second term of Eq.~(\ref{eqn_excess_work_integral}) as a functional of $u$ and $\dot{u}$, which can be minimized using the \mbox{Euler--Lagrange} equation:
\begin{equation}
\label{eqn_u_ddot_equals_zero}
  \frac{d^2u}{dt^2} = 0.
\end{equation}
Thus, excess work is minimized when $u$, the centre of mass of the sample distribution, moves at a constant speed 
\begin{equation}
\label{eqn_u_of_t}
  u(t) = u_0 + mt
\end{equation}
for some constant $m$.  For the protocol Eq.~(\ref{eqn_u_of_t}), the first term on the right-hand side of Eq.~(\ref{eqn_excess_work_integral}) vanishes; therefore, the excess work produced by this protocol is minimized, assuming the parameter $\theta(t)$ changes continuously.

Solving Eq.~(\ref{eqn_u_eom}) for $\theta$ and substituting $\dot{u} = m$, it follows that $\theta(t)$ also moves at the same constant speed:
\begin{equation}
\label{eqn_theta_of_t}
   \theta(t) = \frac{m}{\mu} + u_0 + mt.
\end{equation}
Note that Eq.~(\ref{eqn_u_eom}) becomes $m=\mu(\theta_0 - u_0)$, so if $p_\text{s}$ is initially in equilibrium, that is, $u_0 = \theta_0$, then $m=0$, implying that $p_\text{s}$ and $p_\text{m}$ do not move. Therefore, we need to assume $u_0 \neq \theta_0$, which is standard practice in EBM training, where the model parameters are not usually initialized in equilibrium.

Determining whether it is possible to train the harmonic trap EBM with minimal excess work amounts to finding a time-dependent learning rate schedule $\eta(t)$ which induces the dynamics expressed by Eqs.~(\ref{eqn_u_of_t}) and~(\ref{eqn_theta_of_t}). To find such an optimal learning rate schedule, we substitute Eqs.~(\ref{eqn_u_of_t}) and~(\ref{eqn_theta_of_t}) into Eq.~(\ref{eqn_harmonic_trap_MLE}) and use $\theta^\ast = \theta(\tau)$, which yields the learning rate schedule
\begin{equation}
\label{eqn_optimal_learning_rate}
  \eta(t) = \frac{1}{\tau - t + 1/\mu}.
\end{equation}
This equation represents the time dependence that the learning rate must obey in order to minimize the excess work. The initial condition $\theta(0) = \theta_0$ implies that
\begin{equation}
\label{eqn_mu_tau_condition}
    \mu \tau = \frac{\theta^\ast - \theta_0}{\theta_0 - u_0}.
\end{equation}
This condition has two important consequences. First, since $\mu$ and $\tau$ are strictly positive, there are only two initial parameter orderings which admit solutions: $\theta^\ast < \theta_0 < u_0$ and $u_0 < \theta_0 < \theta^\ast$.  The other four possible orderings have no solution.  
\begin{figure}
  \includegraphics[width=0.48\textwidth]{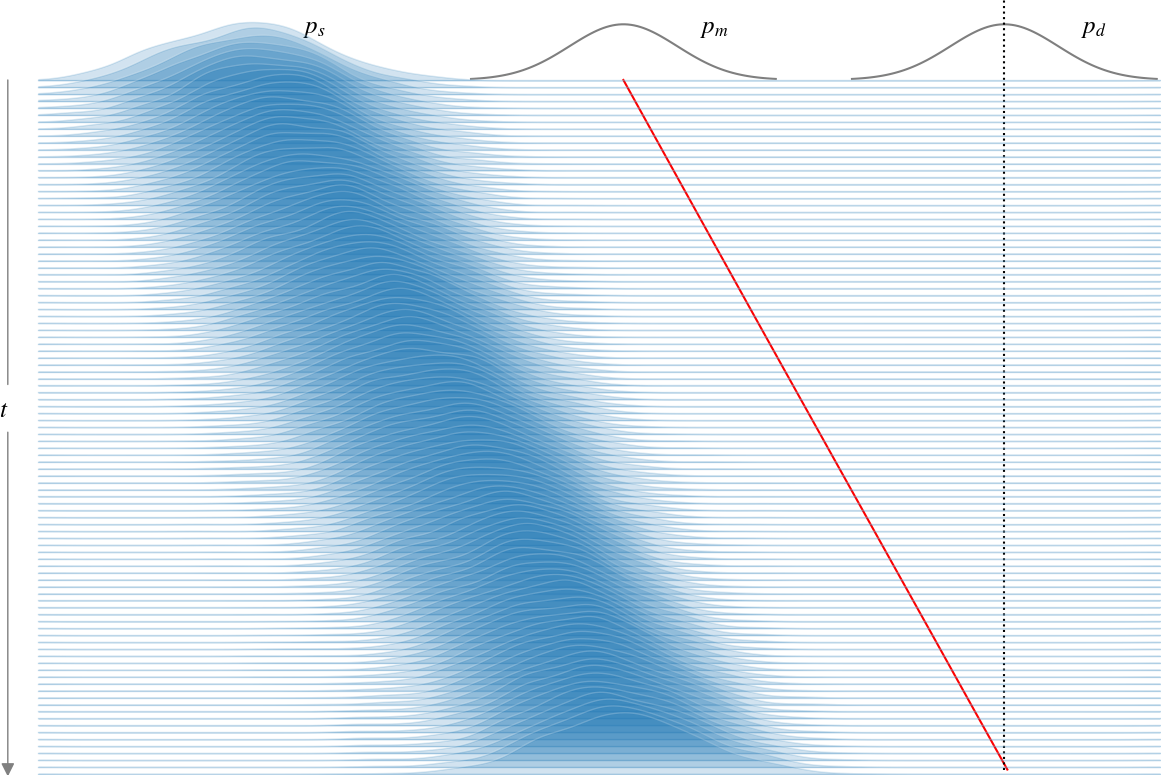}
\caption{The harmonic trap EBM is trained with minimal entropy production using a continuous protocol, i.e., $\theta(t)$ (represented by a red curve) is continuous.  However, this method does not strictly learn the value of $\theta^\ast$ (dotted line), as the value of $\theta^\ast$ must be known a priori to prepare $p_\text{s}$ and $p_\text{m}$ with the specific initial conditions required to minimize excess work. The parameters chosen are $\theta_0=-5$, $\theta^\ast=0$, $u_0=-10$, $\mu=1$, $\beta=1$, and $\tau=1$.}
\label{fig_harmonic_trap_optimal}
\end{figure}
The second consequence of Eq.~(\ref{eqn_mu_tau_condition}) is that, while the excess work is minimized, the parameter $\theta^\ast$ is not strictly learned, since knowledge of its value is required to set $u_0$ and $\theta_0$. In the section that follows, we show that, by allowing for discontinuities in the training protocol, it is actually possible to learn the value of $\theta^\ast$ with minimal excess work.

Figure~\ref{fig_harmonic_trap_optimal} illustrates the effect of using the learning rate schedule~(\ref{eqn_optimal_learning_rate}) with the admissible initial conditions determined by the Eq.~(\ref{eqn_mu_tau_condition}): the distributions move in straight lines and the model parameter stops precisely at $\theta^\ast$.  

\subsubsection{Discontinuous Protocols}

In Ref.~\cite{schmiedl2007optimal}, it was shown that by allowing the protocol $\theta(t)$ to be discontinuous at its endpoints, the excess work can be minimized further than in the case of its analogous continuous protocol.

For a discontinuous protocol, $\theta(t)$ is a piecewise-defined function for which $\theta(0)=\theta_0 \neq \frac{m}{\mu}+u_0$ and $\theta(\tau)=\theta^\ast \neq \frac{m}{\mu}+u_0+m\tau$, while Eq.~(\ref{eqn_theta_of_t}) still holds for $0<t<\tau$.  Consequently, the relation $\dot{u}=m$ holds within the integration interval of Eq.~(\ref{eqn_excess_work_integral}), namely for $0<t<\tau$, but not at the endpoints, where Eq.~(\ref{eqn_u_eom}) should be used instead.  The excess work is thus generalized to
\begin{equation}
\label{eqn_W_ex_discontinuous}
  \langle W_\text{ex} \rangle_{p_\text{s}} = \frac{m^2 \tau}{\mu} + \frac{1}{2} \left[\left(\theta^\ast - u_0 - m\tau\right)^2 - \left(\theta_0 - u_0\right)^2 \right].
\end{equation}
Varying this expression with respect to $m$, its optimal value  is determined to be
\begin{equation}
\label{eqn_m_opt}
    m_\text{opt} = \frac{\theta^\ast - u_0}{\tau + 2/\mu}.
\end{equation}
Following calculations similar to those that lead to Eq.~(\ref{eqn_optimal_learning_rate}), along with $m_\text{opt}$, results in a slightly different optimal learning rate schedule which is valid for $0 < t < \tau$:
\begin{equation}
\label{eqn_learning_rate_discontinuous}
    \eta_\text{opt}(t) = \frac{1}{\tau - t + 2/\mu}.
\end{equation}

At $t=0$ and $t=\tau$, impulses in $\eta(t)$ are required in order to make the jumps in $\theta$. According to Eq.~(\ref{eqn_harmonic_trap_MLE}), the impulse at $t=0$ must have the magnitude
\begin{align}
  \eta_0(t) 
  &= -\frac{\theta(0^+) - \theta_0}{\partial_\theta \mathcal{L}_\text{approx}(\theta)|_{t=0}} \delta(t), \\
  &= \left(\frac{u_0 - \theta_0}{\theta^\ast - u_0} + \frac{1}{\mu \tau + 2}\right) \delta(t),
\label{eqn_impulse_start}
\end{align}
using $\partial_\theta \mathcal{L}_\text{approx}(\theta) = u - \theta^\ast$, $\theta(0^+) = m_\text{opt}/\mu + u_0$, and where $\delta(t)$ denotes the Dirac delta distribution.

While knowledge of $\theta^\ast$ is required to compute this impulse for general initial conditions, notice that it is not required if the sample distribution is initialized in equilibrium, that is $u_0 = \theta_0$.

The impulse at $t=\tau$ is found similarly:
\begin{align}
  \eta_\tau(t) 
  &= -\frac{\theta^\ast - \theta(\tau^-)}{\partial_\theta \mathcal{L}_\text{approx}(\theta)|_{t=\tau}} \delta(t - \tau), \\
  &= \frac{1}{2} \delta(t - \tau).
\label{eqn_impulse_end}
\end{align}
In Figure~\ref{fig_harmonic_trap_optimal_jumps}, the harmonic trap EBM is initialized in equilibrium, that is, $u_0 = \theta_0$, and trained using the learning rate schedule in Eq.~(\ref{eqn_learning_rate_discontinuous}) with the impulses defined by Eqs.~(\ref{eqn_impulse_start}) and~(\ref{eqn_impulse_end}) at the initial and final points of the schedule. 
\begin{figure}[h]
  \includegraphics[width=0.48\textwidth]{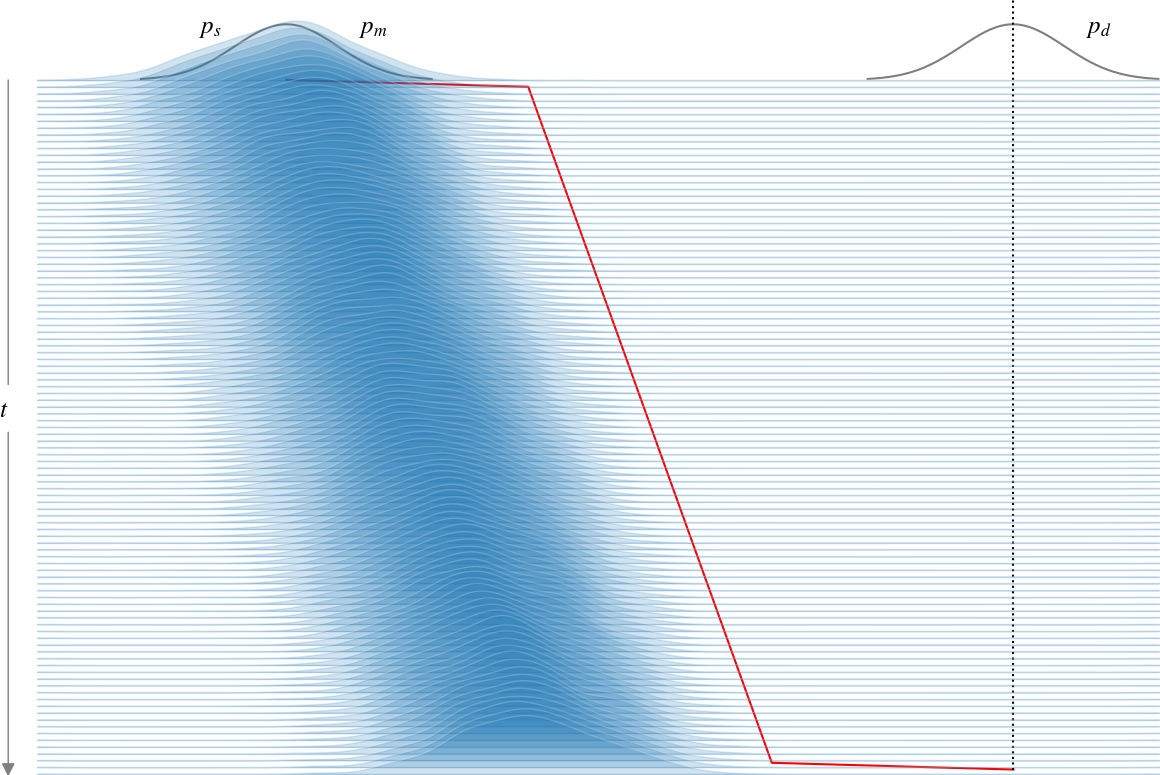}
\caption{The harmonic trap EBM is trained with minimal excess work, using a discontinuous protocol, i.e., the value of $\theta(t)$ (represented by the red curve) has jumps at the endpoints of the schedule. By initializing the system in equilibrium, i.e., $u_0 = \theta_0$, the value of the parameter $\theta^\ast$ (dotted line) is not needed to be known a priori, and so it is genuinely learned.  The parameters chosen are $\theta_0=-10$, $\theta^\ast=0$, $u_0=-10$, $\mu=1$, $\beta=1$, and $\tau=1$.}
\label{fig_harmonic_trap_optimal_jumps}
\end{figure}

If the harmonic trap is initialized in equilibrium, and trained using the discontinuous protocol, then knowing the value of $\theta^\ast$ is not required a priori. This shows that it is possible to learn an unknown target distribution using a persistent chain EBM with minimal excess work.

\subsubsection{Quasi-static Protocols}

In this section, the total excess work for the optimal learning rate schedule is compared to that of a quasi-static process, for which $p_\text{s} \approx p_\text{m}$ for all times, but run for the same finite time $\tau$. We show that, even though the excess work for the quasi-static process approaches zero in the limit of infinite time, the optimal protocol will produce less excess work in a finite time, under the requirement that the quasi-static process is close to convergence after the training process has completed.

We consider a harmonic trap EBM with constant learning rate $\eta \ll \mu$ as an example of a {\textit quasi-static EBM}, that is, one for which the training process is quasi-static. To calculate $W_\text{ex}$ for the harmonic trap EBM with constant learning rate $\eta$, Eqs.~(\ref{eqn_u_eom}) and~(\ref{eqn_harmonic_trap_MLE}) are substituted into Eq.~(\ref{eqn_excess_work_theta_minus_u}) to obtain
\begin{align}
  \langle W_\text{ex} \rangle_{p_\text{s}} 
  &= \int_{0}^{\tau} (\theta - u) \dot{\theta} dt, \\
  &= \int_{0}^{\tau} \frac{-\ddot{\theta} \dot{\theta}}{\eta\mu} dt, \\
  &= -\frac{\dot{\theta}^2}{2\eta \mu}  \bigg|_{0}^{\tau}, \\
  &= -\frac{\eta}{2\mu} (u - \theta^\ast)^2 \bigg|_{0}^{\tau}, \\
  &= \frac{\eta}{2\mu}\left[(u_0 - \theta^\ast)^2 - (u(\tau) - \theta^\ast)^2 \right]. \label{eqn_W_ex_const_eta}
\end{align}  
To calculate $u(\tau)$, we eliminate $\theta$ from Eqs.~(\ref{eqn_u_eom}) and~(\ref{eqn_harmonic_trap_MLE}) to obtain the following second-order ordinary differential equation with constant coefficients:
\begin{equation}
    \ddot{u} + \mu \dot{u} + \eta \mu (u - \theta^\ast) = 0.
\end{equation}
The solution to this equation in the limit $\eta \ll \mu$ is given by
\begin{equation}
    u(t) = \theta^\ast - (\theta_0 - \theta^\ast) e^{-\eta t}.
    \label{eqn_theta_mu_large}
\end{equation}
Finally, substituting Eq.~(\ref{eqn_theta_mu_large}) into Eq~(\ref{eqn_W_ex_const_eta}), we obtain the excess work for the quasi-static harmonic trap:
\begin{equation}
  \langle W_\text{ex}^\text{qs} \rangle_{p_\text{s}} = \frac{\eta}{2\mu}\left[(u_0 - \theta^\ast)^2 - (\theta_0 - \theta^\ast)^2 e^{-2\eta \tau} \right].
\end{equation}
We compare $\langle W_\text{ex}^\text{qs} \rangle_{p_\text{s}}$ with the excess work produced by the optimal finite time protocol, which can be found by substituting Eq.~(\ref{eqn_m_opt}) into Eq.~(\ref{eqn_W_ex_discontinuous}):
\begin{equation}
  \langle W_\text{ex}^\text{opt} \rangle_{p_\text{s}} =  \frac{(\theta_0 - \theta^\ast)^2}{\mu\tau + 2}.
\end{equation}
Taking the ratio of these and assuming $u_0 = \theta_0$,
\begin{equation}
    \frac{\langle W_\text{ex}^\text{qs} \rangle}{\langle W_\text{ex}^\text{opt} \rangle} = \left(\frac{1}{2} + \frac{1}{\mu\tau} \right) \eta\tau \left(1 - e^{-2 \eta \tau}\right).
    \label{eqn_qs_opt_ratio}
\end{equation}
If this ratio is greater than one, then the excess work for this quasi-static EBM is greater than that of the optimal protocol. This will be the case if $\eta \tau$ is not small, more specifically $\eta \tau > 2.24$ for all $\mu \tau > 0$. Note that Eq.~(\ref{eqn_theta_mu_large}) implies that $e^{-\eta t}$ is a measure of how close $u(t)$ is to convergence. Therefore, if we require this quasi-static EBM to converge, that is, $e^{-\eta \tau} < \epsilon$ for some $\epsilon > 0$, then $\eta \tau > 2.24$ implies $\epsilon < 0.10$.  In other words, for $\epsilon < 0.10$, which we take to be a reasonable definition of convergence, the optimal protocol will produce less excess work than this quasi-static EBM.

In summary, a quasi-static process that is close to convergence will produce more excess work than the optimal protocol. Note, in particular, that an EBM trained using the exact MLE gradient given by Eq.~(\ref{eqn_MLE_gradient}) does not minimize the excess work in a finite amount time, for some definition of convergence.

\section{A Learning Rate for General Potentials}
\label{sec_general}

In this section, we derive a generalization of the learning rate schedule given by Eq.~(\ref{eqn_optimal_learning_rate}), for the case of general energy functions. While the motivation for constructing this protocol was to prove that minimal dissipation learning is possible, the main motivation in this section is to provide a physical intuition for the quantities appearing in Eq.~(\ref{eqn_optimal_learning_rate}).  To do so we make the assumption of slow driving, that is, the model updates are much slower than the Langevin dynamical updates.  We leave the extension to the discontinuous protocol Eq.~(\ref{eqn_learning_rate_discontinuous}) for future study.

The excess work produced by an isothermal process can be calculated perturbatively in the total time of the process as is done in Ref.~\cite{sekimoto_Complementarity_1997, wadia2022solution}:
\begin{equation}
    \langle W \rangle_{p_\text{s}} = \Delta F_{eq} + \frac{1}{\tau} \int_{0}^{1} dt \, \dot{\boldsymbol{\theta}}^\text{T} \boldsymbol{\zeta}(\boldsymbol{\theta}) \, \dot{\boldsymbol{\theta}} + \mathcal{O}(\tau^{-2}),
\end{equation}
where $\boldsymbol{\zeta}$ is a symmetric positive-definite matrix referred to as the \textit{thermodynamic metric}.  Therefore, in the limit of slow driving, that is, $\tau \rightarrow \infty$, the gradient can be expressed as
\begin{equation}
\label{eqn_d_W_ex_dt}
    \nabla_{\boldsymbol{\theta}}\langle W_\text{ex} \rangle_{p_\text{s}} \approx \boldsymbol{\zeta}(\boldsymbol{\theta}) \, \dot{\boldsymbol{\theta}}.
\end{equation}
The matrix elements of $\boldsymbol{\zeta}$ can be determined by Green's functions of the Fokker--Planck equation  \cite{sekimoto_Complementarity_1997, wadia2022solution} or using linear response theory~\cite{sivak2012thermodynamic}.  

As is remarked in Ref.~\cite{sivak2012thermodynamic}, the metric $\boldsymbol{\zeta}$ can be decomposed as the Hadamard (i.e., element-wise) product 
\begin{equation}
\label{eqn_integral_relaxation_fisher}
    \boldsymbol{\zeta} = \beta^{-1} \left(\boldsymbol{\tau}^\text{r} \odot \boldsymbol{g} \right),
\end{equation}
where $\odot$ is the Hadamard product and $\boldsymbol{g}$ is the Fisher metric with the matrix elements
\begin{equation}
\label{eqn_fisher_metric}
    g_{ij} \equiv \left\langle \left( \frac{\partial \log p_\text{m}}{\partial \theta_i}\right)\left( \frac{\partial \log p_\text{m}}{\partial \theta_j}\right) \right\rangle_{p_\text{m}}.
\end{equation}
 Here, $\boldsymbol{\tau}^\text{r}$ is the \textit{integral relaxation time}, which is the characteristic correlation time between the thermodynamic forces $\nabla_{\boldsymbol{\theta}} \log p_\text{m}$~\cite{sivak2012thermodynamic}.

We generalize the learning rate schedule Eq.~(\ref{eqn_optimal_learning_rate}) by promoting $\eta$ from a scalar in Eq.~(\ref{eqn_approx_descent}) to the matrix $\boldsymbol{\eta}$ taking the form
\begin{align}
\label{eqn_eta_general}
    \boldsymbol{\eta}^{-1}(\boldsymbol{\theta}, t) 
    &= \beta^{-1}\Big(\tau - t + \boldsymbol{\tau}^{\text{r}}(\boldsymbol{\theta})\Big) \odot \boldsymbol{g}(\boldsymbol{\theta}).
\end{align}
Since $\boldsymbol{g}$ and $\boldsymbol{\zeta}$ are both symmetric and positive definite, $\boldsymbol{\eta}$ is also symmetric and positive definite.  

The following theorem shows that, if the learning rate schedule Eq.~(\ref{eqn_eta_general}) is used to train a persistent chain EBM, then the parameters are optimized according to a natural gradient flow, which has been proven useful in practical applications~\cite{amari1998natural}.
\begin{theorem}
Let $p_\mathrm{s}$ be the distribution of a system obeying the overdamped Langevin Eq.~(\ref{eqn_langevin}) initialized in equilibrium.  Let the parameters $\boldsymbol{\theta}(t)$ of the potential $E(\boldsymbol{x}, \boldsymbol{\theta}(t))$ be driven {\it slowly} according to the MLE gradient flow Eq.~(\ref{eqn_approx_descent}) for a time $\tau$ such that the approximation Eq.~(\ref{eqn_d_W_ex_dt}) always holds.  If the learning rate $\eta$ follows the time-dependent schedule Eq.~(\ref{eqn_eta_general}), then
\begin{equation}
\label{eqn_natural_gradient}
    \dot{\boldsymbol{\theta}} = -\frac{\beta}{\tau - t} \,\boldsymbol{g}^{-1} \, \nabla_{\boldsymbol{\theta}} \mathcal{L}_{\mathrm{MLE}},
\end{equation}
which is a natural gradient flow.
\end{theorem}
\begin{proof}
Substituting Eq.~(\ref{eqn_bias_MLE_approx}) into Eq.~(\ref{eqn_approx_descent}), we have
\begin{equation}
    \dot{\boldsymbol{\theta} }
    = -\boldsymbol{\eta} \nabla_{\boldsymbol{\theta}} \Big( \mathcal{L}_\text{MLE} - \mathcal{L}_\text{bias} \Big).
    \label{eqn_grad_bias_slow_driving}
\end{equation}
Under the assumptions of slow driving and initialization in equilibrium, the approximation Eq.~(\ref{eqn_d_W_ex_dt}) is valid, implying that
\begin{equation}
    \nabla_{\boldsymbol{\theta} }\mathcal{L}_\text{bias} = \nabla_{\boldsymbol{\theta}} \langle W_{\text{ex}} \rangle_{p_\text{s}} = \boldsymbol{\zeta} \dot{\boldsymbol{\theta}}.
\end{equation} 
Substituting this equation into Eq.~(\ref{eqn_grad_bias_slow_driving}) and rearranging gives
\begin{equation}
\nabla_{\boldsymbol{\theta}} \mathcal{L}_\text{MLE}
= -(\boldsymbol{\eta}^{-1} - \boldsymbol{\zeta})\dot{\boldsymbol{\theta}}.
\label{eqn_theta_dot_proof}
\end{equation}
Finally, substituting Eqs.~(\ref{eqn_integral_relaxation_fisher}) and~(\ref{eqn_eta_general}) into Eq.~(\ref{eqn_theta_dot_proof}) yields
\begin{equation}
\nabla_{\boldsymbol{\theta}} \mathcal{L}_\text{MLE} = -\beta^{-1} \left(\tau - t \right) \boldsymbol{g} \,\dot{\boldsymbol{\theta}},  \label{eqn_omega_g_theta_dot}
\end{equation}
which implies Eq.~(\ref{eqn_natural_gradient}).
\end{proof}
In Ref.~\cite{fujiwara1995gradient}, it is shown that for dually flat statistical manifolds, a natural gradient flow of the MLE objective always converges to the global objective in a finite time.  The standard dually flat statistical manifolds are the exponential family and the mixture family~\cite{amari2016information}.  The harmonic trap is a member of the \textit{exponential family}, which is the set of Boltzmann distributions given by Eq.~(\ref{eqn_boltzmann}) having energy functions of the form
\begin{equation}
\label{eqn_exponential_family}
    E(\boldsymbol{x}, \boldsymbol{\theta}) = C(\boldsymbol{x}) + \boldsymbol{\theta} \cdot \boldsymbol{F}(\boldsymbol{x}),
\end{equation}
where $C$ and $\boldsymbol{F}$ are arbitrary functions.

Further, in Ref.~\cite{fujiwara1995gradient}, it is shown that the trajectories of Eq.~(\ref{eqn_natural_gradient}) for dually flat manifolds are geodesics with respect to the Chentsov--Amari connection of the Fisher metric.  This means that, in order for the trajectories of Eq.~(\ref{eqn_natural_gradient}) to minimize excess work, they must also be geodesics with respect to the Levi--Civita connection of the thermodynamic metric $\boldsymbol{\zeta}$. In the section that follows, we show how the geodesics coincide in the case of the harmonic trap.  

\subsection{A Slow Driving Protocol for the Harmonic Trap}

In this section, we use Eq.~(\ref{eqn_eta_general}) for the harmonic trap to compare its excess work to that of other learning rate schedules. To this end, the integral relaxation time $\tau^{\text{r}}$ and the Fisher metric $\boldsymbol{g}$ are required.  These matrices can be inferred from Eq.~(\ref{eqn_integral_relaxation_fisher}), given the thermodynamic metric $\boldsymbol{\zeta}$, which is a well-known mathematical result for the harmonic trap~\cite{sivak2012thermodynamic, wadia2022solution}:
\begin{equation}
\label{eqn_zeta_mu}
    \zeta = \frac{1}{\mu}.
\end{equation}
The Fisher metric can be calculated using Eq.~(\ref{eqn_fisher_metric}) as
\begin{equation}
\label{eqn_g_beta}
    g = \beta.
\end{equation}
It follows that Eq.~(\ref{eqn_integral_relaxation_fisher}) implies the integral relaxation time is
\begin{equation}
\label{eqn_tau_c_mu}
    \tau^\text{r} = \frac{1}{\mu}.
\end{equation}
Note that substituting Eqs.~(\ref{eqn_g_beta}) and~(\ref{eqn_tau_c_mu}) into Eq.~(\ref{eqn_eta_general}) gives the optimal learning rate for the continuous protocol Eq.~(\ref{eqn_optimal_learning_rate}).  

The difference between the protocols Eq.~(\ref{eqn_optimal_learning_rate}) and Eq.~(\ref{eqn_eta_general}) is that the latter is initialized in equilibrium, whereas the former must be initialized in a precise, out-of-equilibrium state which requires prior knowledge of the value of $\theta^\ast$.

Since the thermodynamic metric Eq.~(\ref{eqn_zeta_mu}) is constant, the Levi--Civita geodesics for the parameter $\theta(t)$ are straight lines.  Note that this is not equivalent to, and in this case more general than, the condition specified by Eq.~(\ref{eqn_u_ddot_equals_zero}). Substituting $\theta(t) = \theta_0 + mt$ into Eq.~(\ref{eqn_u_eom}) and solving the resulting differential equation yields
\begin{equation}
\label{eqn_u_of_t_general}
    u(t) = \frac{m}{\mu}e^{-\mu t} + \theta_0 - \frac{m}{\mu} + mt,
\end{equation}
which is more general than Eq.~(\ref{eqn_u_of_t}). In fact, since Eq.~(\ref{eqn_u_of_t_general}) does not satisfy Eq.~(\ref{eqn_u_ddot_equals_zero}), it does not minimize the total entropy production, due to the slow driving approximation. However, since the initialization in equilibrium does not depend on $\theta^\ast$, it is fair to say that $\theta^\ast$ is genuinely learned.

In Fig.~\ref{fig_harmonic_trap_thermo}, the training of the harmonic trap is illustrated for the slow diving protocol Eq.~(\ref{eqn_u_of_t_general}). Note that the sample distribution (shown using blue ridge plots) indeed does not follow a straight path whereas the model parameter (indicated by a red curve) does.
\begin{figure}[h]
  \includegraphics[width=0.48\textwidth]{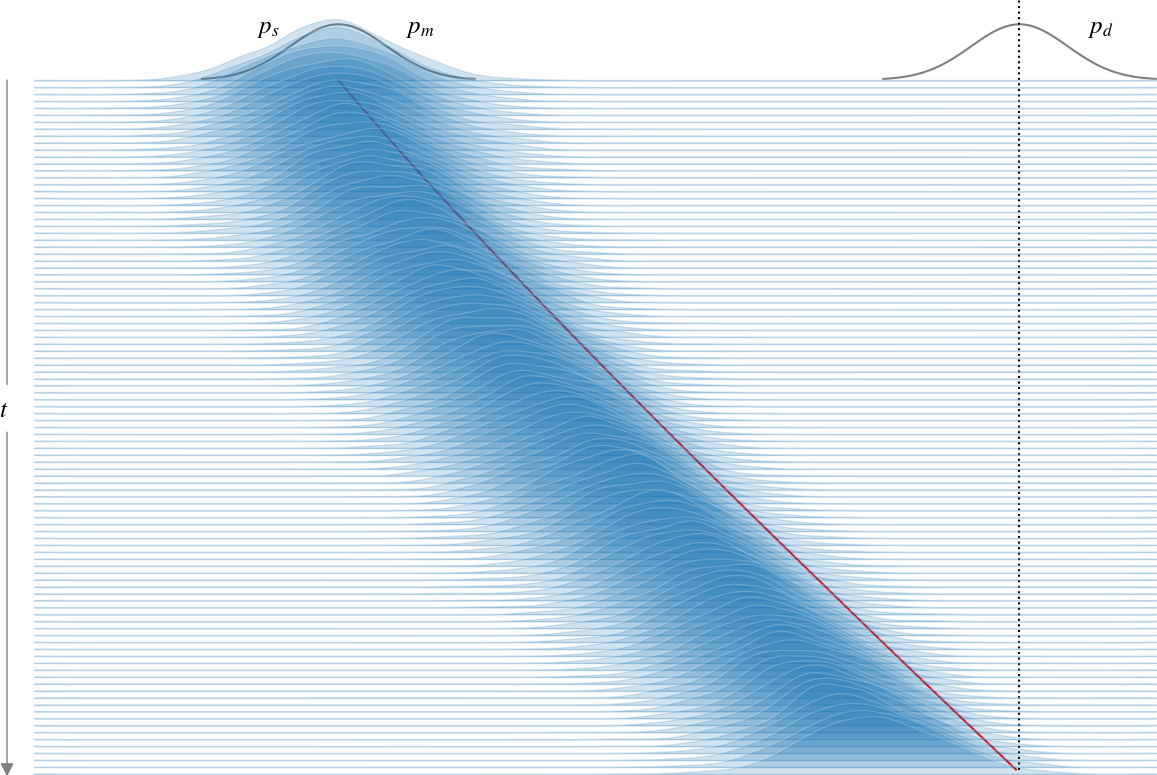}
\caption{The harmonic trap EBM is trained using the slow driving protocol Eq.~(\ref{eqn_eta_general}), which follows geodesics with respect to the thermodynamic metric.  The parameters chosen are $\theta_0=-10$, $\theta^\ast=0$, $u_0=-10$, $\mu=5$, $\beta=1$, and $\tau=1$. }
\label{fig_harmonic_trap_thermo}
\end{figure}

The excess work for the slow driving protocol Eq.~(\ref{eqn_eta_general}) can be calculated using Eq.~(\ref{eqn_excess_work_theta_minus_u}) to be
\begin{equation}
\label{eqn_excess_work_slow}
    \langle W_\text{ex}^{\text{slow}} \rangle_{p_\text{s}} = \frac{m^2 \tau}{\mu} - \frac{m^2}{\mu^2}(1 - e^{-\mu \tau}).
\end{equation}
The first term on the right-hand side is the excess work of the continuous (i.e., minimal-entropy) protocol in Eq.~(\ref{eqn_theta_of_t}).  The second term is strictly negative, and so it decreases the excess work, meaning that the slow driving protocol, at least in the case of the harmonic trap, always produces less excess work than the continuous protocol.  Taking the ratio of the left-hand and right-hand sides of Eqs.~(\ref{eqn_excess_work_slow}) and~(\ref{eqn_W_ex_discontinuous}) gives
\begin{equation}
    \frac{\langle W_\text{ex}^\text{slow} \rangle}{\langle W_\text{ex}^\text{opt} \rangle} = \frac{\mu\tau + 2}{(\mu\tau)^2}\left(\mu \tau - 1 + e^{-\mu\tau}\right),
\end{equation}
which is always greater than $1$ for $0 < \mu \tau < \infty$.  Therefore, the excess work produced by the slow driving protocol is always greater than that of the discontinuous (i.e., optimal) protocol.  Figure~\ref{fig_excess_work} shows a plot of the cumulative excess work versus time for each of the protocols we consider for some chosen parameters.
\begin{figure}[h]
  \includegraphics[width=0.48\textwidth]{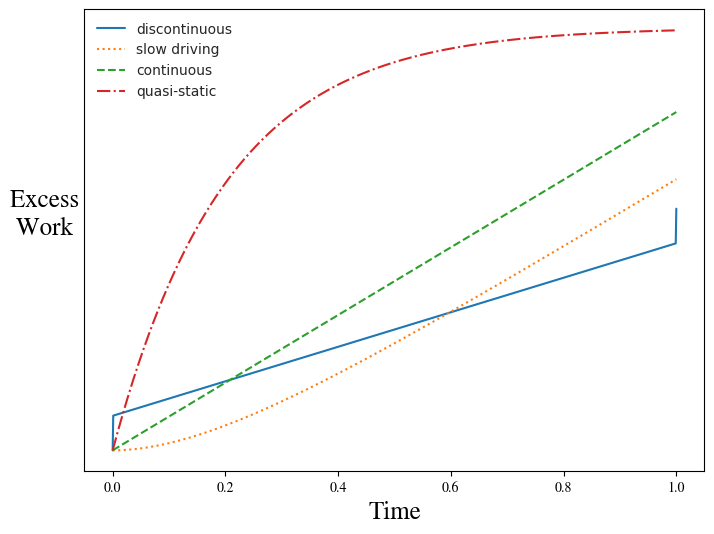}
\caption{Cumulative excess work for each of the protocols plotted versus the normalized time $t/\tau$.  The protocols correspond to the figures above as follows: discontinuous to Fig.~\ref{fig_harmonic_trap_optimal_jumps}, slow driving to Fig.~\ref{fig_harmonic_trap_thermo}, continuous to Fig.~\ref{fig_harmonic_trap_optimal}, and quasi-static to Fig.~\ref{fig_harmonic_trap_quasi_static}.  The parameters used are $\theta^\ast - \theta_0 = 1$, $\mu=1$, and $\tau=5$.  For the quasi-static protocol, $\eta=0.5$.}
\label{fig_excess_work}
\end{figure}
%
%------------------------------------------------
\section{Discussion}

By viewing persistent chain EBMs~\cite{nijkamp2020anatomy} through the lens of stochastic thermodynamics, we have found that the bias of the MLE objective is exactly equal to the thermodynamic excess work. This excess work, and hence the bias, is known to be minimized by a quasi-static process, that is, one for which the sample distribution is always close to equilibrium; however, the excess work will vanish only in the limit of infinite time. For a discrete simulation of such a process, this limit corresponds to performing a very large number of MCMC iterations, which would be computationally impractical.

This led to the question of whether it is possible to train an EBM that produces a minimal amount of excess work in a finite amount of time. Furthermore, to respect the MLE objective, we assumed that the only controllable parameter is the learning rate. We showed that training a harmonic trap EBM in a finite amount of time with minimal excess work is indeed possible. We found that, if the protocol for the model parameter is required to be continuous, then it is not possible to learn the target distribution, since doing so requires precise initial conditions that depend on the target distribution. However, if the protocol for the model parameter is allowed to have discontinuities at the initial and final points as in Ref.~\cite{schmiedl2007optimal}, we found that the target distribution can be learned if  the sample distribution is initialized in equilibrium.  We leave the extension of our results to the application of optimal control and statistical physics in supervised machine learning~\cite{mignacco2025statistical} to future study.

In addition, we proposed a generalization of minimal-dissipation learning rate schedules for a harmonic trap to general energy potentials. This generalization is based on the thermodynamic metric proposed in Ref.~\cite{sivak2012thermodynamic}, and while it does minimize excess work for the harmonic trap, we did not determine the broader class of distributions for which excess work is minimized. Interestingly, we found that it induces a natural gradient flow on the MLE objective, which is a well-known second-order optimization method~\cite{martens2020new}.  We leave the determination of the class of distributions which learn with minimal dissipation using this learning rate schedule for future work. We note that for multi-modal distributions, the slow driving approximation can be violated if the relaxation time diverges, and we leave this this for future investigations.

We note that the continuous-time Langevin dynamical systems considered here correspond to ideal thermodynamic systems, and the excess work that they produce is a lower bound on the energy dissipated by any physical device used to simulate them. For the discretized algorithms of Langevin dynamics, the thermodynamic excess work must also take into account additional degrees of freedom, such as memory, which is an avenue of research we leave for future research. 

\section*{Acknowledgements}

We thank Marko Bucyk for his careful review and editing of the manuscript. We thank Pooya Ronagh, Mark Schmidt, Sayonee Ray, Alan Milligan, Enrico Olivucci, Maximillian Puelma Touzel, and Yvonne Geyer for useful discussions. We also thank the referees for the careful reading and useful suggestions. This work was supported by the National Research Council of Canada (NRC) under AQC-206-2.

\bibliography{apssamp}% Produces the bibliography via BibTeX.

\end{document}